\documentclass[11pt]{article}

\usepackage{amsmath}
\usepackage{amsthm}
\usepackage{amssymb}
\usepackage{algorithm}
\usepackage{color}
\usepackage{graphicx}
\usepackage{grffile}
\usepackage{wrapfig,epsfig}
\usepackage{epstopdf}
\usepackage{multirow}
\usepackage{url}
\usepackage{color}
\usepackage{epstopdf}
\usepackage{algpseudocode}
\usepackage[T1]{fontenc}
\usepackage{bbm}
\usepackage{comment}
\usepackage{dsfont}
\usepackage{thm-restate}
\usepackage{bm}
\usepackage{tcolorbox}
\usepackage{subcaption}
\usepackage{enumitem}
\usepackage[margin=1in]{geometry}
\usepackage{mathtools}

\newtheorem{theorem}{Theorem}[section]
 
\newtheorem{lemma}[theorem]{Lemma}

\newtheorem{fact}[theorem]{Fact}

\theoremstyle{definition}
\newtheorem{remark}[theorem]{Remark}
\newtheorem{definition}[theorem]{Definition}

\newcommand{\N}{\mathcal{N}}
\newcommand{\mI}{\mathcal{I}}

\newcommand{\eps}{\epsilon}
\newcommand{\wt}{\widetilde}
\newcommand{\TV}{\mathsf{TV}}
\newcommand{\EFG}{\mathsf{EFG}}
\newcommand{\mS}{S^{\mathsf{EFG}}}

\usepackage{complexity}

\renewcommand{\E}{\mathbb{E}}
\DeclareMathOperator*{\argmax}{arg\,max}

\definecolor{DarkBrown}{rgb}{0.4,0.13,0.13}

\makeatletter
\newcommand*{\RN}[1]{\expandafter\@slowromancap\romannumeral #1@}
\makeatother

\title{The complexity of approximate (coarse) correlated equilibrium for incomplete information games}
\author{  
Binghui Peng
\\ Columbia University \\  \texttt{bp2601@columbia.edu} 
\and Aviad Rubinstein
\\ Stanford University \\ \texttt{aviad@stanford.edu}
}
\date{\today}

\begin{document}
\maketitle
\begin{abstract}

We study the iteration complexity of decentralized learning of approximate correlated equilibria in incomplete information games. 

On the negative side, we prove that in {\em extensive-form games}, assuming $\PPAD \not\subset \TIME(n^{\polylog(n)})$, any polynomial-time learning algorithms must take at least $2^{\log_2^{1-o(1)}(|\mI|)}$ iterations to converge to the set of $\eps$-approximate correlated equilibrium, where $|\mI|$ is the number of nodes in the game and $\eps > 0$ is an absolute constant. This nearly matches, up to the $o(1)$ term, the algorithms of \cite{peng2023fast, dagan2023external} for learning $\eps$-approximate correlated equilibrium, and resolves an open question of Anagnostides, Kalavasis, Sandholm, and Zampetakis \cite{anagnostides2024complexity}. Our lower bound holds even for the easier solution concept of $\eps$-approximate {\em coarse} correlated equilibrium

On the positive side, we give uncoupled dynamics that reach $\eps$-approximate correlated equilibria of a {\em Bayesian game} in polylogarithmic iterations, without any dependence of the number of types. This demonstrates a separation between Bayesian games and extensive-form games.

\end{abstract}

\setcounter{page}{0}
\thispagestyle{empty}

\newpage
\section{Introduction}
\label{sec:intro}
Our work considers fundamental questions on modeling and designing agents who learn to interact in settings with incomplete information. From an online learning perspective, we want agents that achieve low {\em external regret}%
\footnote{Definitions of external and swap regret, Bayesian and extensive-form games, and various notions of equilibria are deferred to Section~\ref{sec:pre}.}, ideally also satisfying the stronger notion of low {\em swap regret}. From a game theory perspective, we consider agents that learn to play in games of incomplete information; we ask whether they can efficiently converge to the set of approximate {\em coarse coarse correlated (CCE)} equilibria, ideally also satisfying the stronger notion of approximate {\em correlated equilibria (CE)}. 
The connection between online learning and equilibrium computation is well-known (e.g. see \cite{blum2007external}): If each agent runs an online learning algorithm with low external regret, then the empirical distribution converges to the set of approximate CCE; if it further satisfies low swap regret, then the distribution converges to the set of approximate CE .



Correlated equilibria (specifically, normal-form correlated equilibria, NFCE%
\footnote{Different variants of CCE and CE have been considered for games of incomplete information. In what is arguably the most natural and desired variant, normal-form CE (NFCE), the correlation can be implemented by a correlating device sending each player a signal ex-ante, i.e.~one signal that is independent of their type (in Bayesian game) or node in the game tree (in an extensive-form game). Motivated primarily by the conjecture that NFCE are intractable, several other variants have been considered~\cite{von2008extensive,fujii2023bayes}.}) have long been conjectured intractable~\cite{von2008extensive}, until recently \cite{dagan2023external,peng2023fast} gave algorithms that, for both Bayesian games and more general extensive-form games, compute $\eps$-NFCE in $(|\mI| \log(n))^{O(1/\eps)}$ iterations of uncoupled dynamics, where $|\mI|$ is the size of the game (number of types or game nodes) and $n$ is the number of actions. Put in the language of online learning, agents can achieve $\eps$-swap regret after $(|\mI| \log(n))^{O(1/\eps)}$ iterations. While this is polynomial in the full description of the game, the dependence on $|\mI|$ (which is often combinatorial in natural game parameters) is undesirable.

A natural question is whether the algorithms of~\cite{dagan2023external,peng2023fast} can be improved. Our results show that the answer depends on whether we have a Bayesian game or an extensive-form game:

\paragraph{A lower bound on (poly-time) external regret in extensive-form games}
Recent work of~\cite{anagnostides2024complexity} showed that even with only three players, and even if one settles for the weaker desiderata of low {\em external} regret or approximate {\em coarse} correlated equilibrium, they cannot be achieved in extensive-form games by polynomial-time agents running for fewer than $2^{\log_2^{1/2-o(1)}(|\mI|)}$ iterations, assuming the ``ETH for \PPAD''~\cite{babichenko2016can}.
In this work we give an improved reduction that:
\begin{itemize}
    \item Raises the lower bound on the number of iterations to $2^{\log_2^{1-o(1)}(|\mI|)}$; notice that this is almost tight as $2^{\log_2^{1+o(1)}(|\mI|)}$ is already sufficient by~\cite{dagan2023external,peng2023fast}. 
    \item Weakens the complexity assumption from $\PPAD \not\subset \TIME(2^{\tilde{\Omega}(n)})$ to $\PPAD \not\subset \TIME(n^{\polylog(n)})$; this lower bound is almost tight since if $\PPAD = \FP$ players can simply compute and play a Nash equilibrium. 
\end{itemize}

\begin{restatable}[]{theorem}{ExtensiveLB}
\label{thm:lb-extensive-main}
Let $\eps > 0$ be an absolute constant. Assuming $\PPAD \not\subset \TIME(n^{\polylog(n)})$, there is no polynomial time learning dynamics that guarantee to converge to the set of $\eps$-CCE of a three player extensive-form game in $2^{\log_2^{1-o(1)}(|\mI|)}$ iterations.
\end{restatable}

This resolves an open question of \cite{anagnostides2024complexity} on the iteration complexity of CCE in extensive-form games, up to the $o(1)$ term in the exponential.

\paragraph{An algorithm for swap regret minimization in Bayesian games}
Our second result shows that for Bayesian games, there is a much better learning dynamics: it runs in time polynomial in the game description, and achieves $\eps$-swap regret (equivalently, $\eps$-CE), using only $(\log(n)/\eps^2)^{O(1/\eps)}$ iterations; in particular there is no dependence whatsoever on the number of types!

\begin{restatable}[]{theorem}{BayesianUB}
\label{thm:bayesian-regret-main}
Let $\eps > 0$ and $n$ be the number of actions.
There is an uncoupled dynamic that converges to the set of $\eps$-CE of a Bayesian game in $(\log(n)/\eps^2)^{O(1/\eps)}$ iterations.
\end{restatable}

This separates the complexity of computing $\eps$-CE in Bayesian and extensive-form games. 

\subsection{A brief overview of our approach}

\subsubsection{Lower bound for extensive-form game}
\label{sec:tech}

We first review the previous lower bound of~\cite{anagnostides2024complexity}, which is based on a framework of~\cite{foster2023hardness}. 
Then we give a high level overview of our new ideas.

\paragraph{Previous approach}

Recent work of~\cite{foster2023hardness} gives a new framework for lower bounding the number of iterations for obtaining approximate CCE; \cite{foster2023hardness} originally used this idea to Markov games, and~\cite{anagnostides2024complexity} very recently applied it to extensive form games. These reductions -as well as ours- are based on three key ideas:

\begin{description}
    \item[Low rank CCE] We say that a correlated distribution is {\em low rank} (called ``sparse'' in previous work%
    \footnote{We believe that ``low rank'' is a better term to distinguish from distributions with sparse supports.}) if it can be written as a mixture of a small number of product distributions. Uncoupled learning dynamics where agents independently update their strategies based on the history of play induce correlated distributions whose rank is at most the number of iterations.     
    The reductions prove hardness results for the problem of computing low-rank approximate CCE; as a corollary, computationally efficient uncoupled dynamics cannot approach approximate CCE / low external regret.  

    \item[Repeated game] At a high level, the reductions in~\cite{foster2023hardness, anagnostides2024complexity} begin with a normal-form game where (approximate) Nash equilibrium is hard to find. The reductions construct an extensive form (or Markov) game which sequentially repeats the hard normal-form game $H$ times. Even though it is easy to find an (approximate) CCE in the normal-form game, recall that it is crucial that each player only observes their signal. However if the CCE is low rank, the players can learn each other's signals from their strategies in previous repetitions, breaking the correlation. 

    \item[Kibitzer player] The downside of sequentially repeating the normal-form games is that it brings spurious equilibria which don't exist in the single-shot game; in particular, equilibria where players may be incentivized to play low-payoff actions with threats of retaliations in future repetitions of the game. To circumvent these artifacts of extensive-form games, we would like to re-normalize the payoff of each player in each repetition with respect to other players' strategies, then players are incentivized to play myopically because any (dis)advantage they could obtain in future repetitions will be normalized away anyway. To correctly compute this normalization, ~\cite{borgs2008myth} introduce an additional player, {\em the Kibitzer}: starting from the original game between Alice and Bob,  in each iteration the Kibitzer competes with one player, say Alice
    ; both Kibitzer and Alice are scored relative to each other's strategy when measured against the same Bob strategies. 
\end{description}

\paragraph{Our approach}
In this work, we improve over the reduction from~\cite{anagnostides2024complexity} in both complexity assumption and lower bound on the number of iterations (in our reduction, both are nearly tight). The key to obtaining a better reduction is that we start from $\eps$-Bayesian Nash equilibrium ($\eps$-BNE\footnote{Specifically, we reduce from what we call {\em (every-type-$\eps$)-BNE}, where each player can, {\em at any type}, gain at most $\eps$, by deviating from its current strategy. A formal definition is shown in Definition~\ref{def:eq}, followed by a detailed discussion.}) for a one-shot Bayesian game instead of $\eps$-NE of a normal-form complete information game.
The advantage of starting the reduction from $\eps$-BNE is that it is a much harder problem than $\eps$-NE: While $\eps$-NE can be found in quasipolynomial time~\cite{lipton2003playing}, $\eps$-BNE is \PPAD-complete~\cite{rubinstein2018inapproximability}. 

The reduction from $\eps$-BNE becomes more delicate because it is possible that even off-equilibrium players only have an $\eps$-improving deviation for one type, while for a random/average type their strategy is approximately optimal%
\footnote{This would not be an issue if we start the reduction from what we call {\em (ex-ante-$\eps$)-BNE}, however this more relaxed solution concept is not known to be \PPAD-hard}. Here we give another application of the Kibitzer player: we let Kibitzer find a type that admits an $\eps$-improving deviation. Specifically, for each repetition of the Bayesian game, the Kibitzer first chooses a player, say Alice. Then Kibitzer chooses a type $\theta$ for Alice. Kibitzer's goal is to outperform Alice, on the type it chooses and the action it plays. Nature (the ``chance'' player) then samples the type $\theta'$ of Bob, conditioning on type $\theta$ for Alice. Crucially, this way Kibitzer does not know $\theta'$.
Finally, Kibitzer and Alice choose actions for type $\theta$, and Bob chooses an action for type $\theta'$.

\subsubsection{Algorithm for Bayesian games}

Our algorithm is based on the swap-to-external regret reduction of \cite{peng2023fast, dagan2023external}. The idea is to run multiple instance of low-external-regret algorithms, each instance restarts and updates in different frequency. \cite{peng2023fast, dagan2023external} prove that as long as the low-external-regret algorithm guarantees at most $\eps$-external regret in $H$ iterations, the overall algorithm has $\eps$-swap regret in $H^{O(1/\epsilon)}$ iterations. In a Bayesian game, if each player has $K$ types and $n$ actions, the entire strategy space has size $n^{K}$. By running MWU over the strategy space, it takes $H = O(\log(n^{K})/\eps^2) = O(K\log(n)/\eps^2)$ iterations to guarantee $\eps$-external regret, and therefore, a total number of $(K\log(n)/\eps^2)^{O(1/\eps)}$ iterations.

Recall the number of iterations in Theorem \ref{thm:bayesian-regret-main} is independent of the number of types $K$. Our key idea is, instead of running MWU over the entire strategy space, we run MWU over each type and the final strategy is taken as the product at each type\footnote{This is equivalent to running MWU over the entire strategy space, but assigning different learning rates to each type}. The key observation is, the appearance probability of each type is fixed during no-regret learning, so after taking the product over other types, the external regret guarantee continues to hold (for every type). One can then apply the swap-to-external regret analysis of \cite{peng2023fast, dagan2023external}.

\subsection{Related work}

\paragraph{Decentralized learning in games} 
There is a long line of work regard learning in games \cite{freund1997decision,freund1999adaptive, foster1997calibrated,foster1998asymptotic,foster1999regret,hart2000simple,hart2001reinforcement,cesa2003potential,stoltz2005internal,blum2007external,stoltz2007learning, papadimitriou2008computing, hart2010long, hart2013simple, syrgkanis2015fast,foster2016learning, chen2020hedging, daskalakis2021near, anagnostides2022near, anagnostides2022uncoupled, peng2023near, peng2023fast, dagan2023external}.

Our work focuses on incomplete information games -- Bayesian games and extensive-form games -- a central topic in the literature \cite{forges1993five, hartline2015no, fujii2023bayes,von2008extensive, farina2022simple, zinkevich2007regret,lanctot2009monte, farina2019efficient, farina2019optimistic, farina2022kernelized, bai2022efficient,anagnostides2023near, peng2023fast, dagan2023external} with important real-world applications \cite{ bowling2015heads, brown2018superhuman,brown2019superhuman,brown2019deep}.
For CCE, there are decentralized learning algorithms (e.g \cite{farina2022kernelized}) that converge to the set of $\eps$-CCE of extensive-form games (resp. Bayesian game) in $O(|\mI|\log(n))$ iterations (for constant $\eps >0$).
Our lower bound (Theorem \ref{thm:lb-extensive-main}) gives a near matching lower bound.
For the stronger notion of NFCE, it was often conjectured intractable (e.g. see \cite{von2008extensive}).
The recent work \cite{peng2023fast,dagan2023external} refute this conjecture for constant $\eps > 0$ and give an algorithm that converges to an $\eps$-NFCE in $(|\mI|\log(n)/\eps^2)^{O(1/\epsilon)}$ iterations. Our work builds upon the framework of \cite{peng2023fast,dagan2023external} but improves the number of iterations to $(\log(n)/\eps^2)^{O(1/\epsilon)}$ for Bayesian games, which are completely independent of the number of types.

\paragraph{Lower bounds for games}
The complexity class of \PPAD~(Polynomial Parity Arguments on Directed graphs) was first introduced by Papadimitriou \cite{papadimitriou1994complexity} to capture a particular genre of total search functions. 
Since then, a variety of important problems were shown to be \PPAD-hard \cite{chen2009settling-market, vazirani2011market, chen2013complexity,othman2016complexity, rubinstein2018inapproximability, rubinstein2019hardness, filos2021complexity, daskalakis2021complexity, papadimitriou2021public,deligkas2022pure, chen2022computational, chen2023complexity,  daskalakis2023complexity, jin2023complexity}, a well-known example is the complexity of Nash equilibrium in normal-form games \cite{chen2009settling,daskalakis2009complexity}.

A recent line of work \cite{foster2023hardness, anagnostides2024complexity} studied the complexity of decentralized learning for coarse correlated equilibrium, by examining the computational complexity of finding a low rank coarse correlated equilibrium. \cite{foster2023hardness} initiated this line of work and showed it is \PPAD-hard to find an $\eps$-CCE in Markov game with polynomial rank, for inverse polynomially small $\eps$\footnote{\cite{foster2023hardness} also showed a lower bound for constant $\eps$ under ETH for \PPAD}.
\cite{anagnostides2024complexity} studied the extensive-form game and proved that there is no polynomial time algorithm for finding rank-$2^{\log_{2}^{1/2-o(1)}|\mI|}$ $\eps$-CCE in extensive-form game for some constant $\eps > 0$, under the exponential time hypothesis (ETH) for \PPAD.
Here ETH for \PPAD~is a refinement of the $\P \neq \PPAD$ conjecture, which asserts that the \PPAD-complete problem {\sc End-of-Line} requires exponential time. This conjecture was motivated by the exponential lower bound in the black box model 
and it was used in the work of \cite{rubinstein2016settling} to rule out polynomial time algorithms for finding approximate Nash equilibrium in normal-form games. 
In our work, we use a substantial weaker conjecture, i.e., we only need to assume the complete problem in \PPAD~does not have quasipolynomial time algorithms.

\section{Preliminary}
\label{sec:pre}

\paragraph{Notation} Throughout the paper, we write $[n] = \{1,2, \ldots, n\}$ and $[n_1: n_2] = \{n_1, n_1+1, \ldots, n_2\}$. For any finite set $S$, we use $\Delta(S)$ to denote probability distributions over $S$. For any distribution $p, q$, we use $\TV(p; q)$ to denote the total variation distance between $p$ and $q$.

\subsection{Bayesian game}
In an $m$-player Bayesian game, let $\Theta_i$ $(|\Theta_i| = K)$ be the set of types and let $A_{i}$ ($|A_i| = n$) be the set of actions for player $i$.
Let $\Theta = \Theta_1 \times \cdots \times \Theta_m$ be the set of type profiles of all players.
The type profile $\theta \in \Theta$ is generated from a prior distribution $\rho \in \Delta(\Theta)$.
Let $A = A_{1} \times \cdots \times A_{m}$ be the action profiles of all players.
The utility function $u_i: \Theta \times A \rightarrow [0,1]$ of player $i$ maps a type profile and an action profile to a payoff value.
A strategy $s_i: \Theta_i  \rightarrow A_i$ determines the action $s_i(\theta_i) \in A_i$ to be selected when player $i$ has private type $\theta_i \in \Theta_i$. 
Let $S_i = A_i^{\Theta_i}$ be the set of pure strategies of player $i$ and $S = S_1 \times \cdots \times S_m$ be the entire strategy space. 
We use $\Delta(S_i)$ to denote the set of mixed strategies, and $X_i := \prod_{\theta_i \in \Theta_i}\Delta(A_i)$ to denote the set of behaviour strategies.


A normal-form correlated equilibrium (also known as strategic-form correlated equilibrium) is a joint distribution $\mu \in \Delta(S)$ over the strategy space, such that no player can gain from deviating from the recommend strategy. Formally 
\begin{definition}[Normal-form correlated equilibrium] A distribution $\mu\in \Delta(S)$ is a normal-form correlated equilibrium if for every player  $i\in [m]$ and all swap function $\phi: S_i\rightarrow S_i$, 
\begin{align*}
\E_{\theta\sim \rho} \E_{s\sim \mu}[u_i(\theta; s(\theta))] \geq \E_{\theta\sim \rho} \E_{s\sim \mu}[u_i(\theta; \phi(s_i)(\theta_i), s_{-i}(\theta_{-i}))].
\end{align*}
\end{definition}

A {\em Bayesian Nash equilibrium (BNE)} is the special case of NFCE where $\mu$ is a product distribution, i.e.~each player samples their action independently. More generally, we consider {\em low-rank} distributions that interpolate between a product distribution (rank 1) and general correlated distribution (unbounded rank).

Two natural ways of extending NFCE and BNE to approximate equilibrium have been considered in the literature; the differences are subtle but important for our purpose. The first notion of approximate equilibrium, which we call {\em ex-ante-approximate}, considers the potential benefit for a player from deviating from the recommended strategy in expectation over the Bayesian prior over types. The second, stricter notion, which we call {\em every-type-approximate} considers the potential benefit for a player from deviating -- maximized over her types and in expectation over all other players' types. This distinction is  similar (but orthogonal) to approximate Nash equilibrium and well-supported Nash equilibrium in normal-form games (e.g.~\cite{daskalakis2009complexity}). 

\begin{definition}[Notions of approximate equilibria] \label{def:eq}\hfill
\begin{itemize}
    \item We say that a distribution $\mu$ is an {\em  (ex-ante-$\eps$)-NFCE} (or {\em (ex-ante-$\eps$)-BNE} when $\mu$ is a product distribution) if for every player  $i\in [m]$ and every swap function $\phi: S_i\rightarrow S_i$,
    \begin{align*}
\E_{\theta\sim \rho} \E_{s\sim \mu}[u_i(\theta; s(\theta))] \geq \E_{\theta\sim \rho} \E_{s\sim \mu}[u_i(\theta; \phi(s_i)(\theta_i), s_{-i}(\theta_{-i}))] - \eps.
\end{align*}

\item We say that a distribution $\mu$ is an {\em (every-type-$\eps$)-NFCE} (or {\em (every-type-$\eps$)-BNE} when $\mu$ is a product distribution) if for every player  $i\in [m]$, every swap function $\phi: S_i\rightarrow S_i$, and for every type $k \in \Theta_i$,
\[
\E_{\theta_{-i} \sim (\rho| \theta_i = k)}\E_{s\sim \mu}[u_i(\theta; s(\theta))] \geq \E_{\theta_{-i} \sim (\rho| \theta_i = k)}\E_{s\sim \mu}[u_i(\theta; \phi(s_i)(\theta_i), s_{-i}(\theta_{-i}))] - \eps.
\]

\end{itemize}

\end{definition}

Our algorithm for Bayesian games gives an every-type-approximate NFCE, which is another improvement over~\cite{daskalakis2023complexity, peng2023fast} that compute ex-ante-approximate NFCE. Meanwhile, our hardness for extensive-form games holds also with respect to the more lenient notion of ex-ante approximation.   
A detailed discussions on different notions of correlated equilibrium in Bayesian games can be found at~\cite{fujii2023bayes}.

\subsection{Extensive-form game} 
The extensive-form game extends Bayesian game by incorporating both incomplete information and sequentiality. An $m$-player extensive-form game can be formulated as a directed game tree $\Gamma$. Let $\N$ be all nodes of $\Gamma$, the nodes of the game tree are partitioned into decision nodes and chance nodes $\N= \N_{1} \cup \cdots \N_{m} \cup \N_{c}$. Here $\N_{i}$ ($i \in [m]$) is the set of nodes where player $i$ takes actions and $\N_{c}$ are chance nodes.
The function of a chance node is to assign an outcome of a chance event, and each outgoing edge represents one possible outcome of that chance event as well as the probability of the event. At a decision node, the edges represent actions and the successor states the result from the player taking those actions.
Player $i$'s decision nodes $\N_i$ are further partitioned into a family $\mI_i$ of information sets; for each information set $I \in \mI_i$, Player $i$ chooses an action from $A_i$.
We consider {\em perfect recall} extensive-form games so that each player remembers the history action sequence leads to each information set.
We use $|\mI| = \sum_{i\in [m]}|\mI_i|$ to denote the total number of information sets.
Let $\mS_i = \prod_{I \in \mI_i} A_{i} = (A_i)^{\mI_i}$ be the set of pure strategies of player $i$, and the entire strategy space is $\mS = \prod_{i\in [m]}\mS_i$.

Let $\Delta(\mS_i)$ be the set of {\em mixed strategies}. Note that this notion of mixed strategies is very general in that it allows a player to correlate their actions across information sets that cannot simultaneously occur in the same play of the game. We also consider a slightly more  restricted notion of {\em behaviour strategies} that do not account for such correlations; formally the set of Player $i$'s behaviour strategies is given by $X^{\EFG}_i = \prod_{I\in \mI_i}\Delta(A_i)$.
Given a strategy distribution $\mu \in \Delta(\mS)$, we use $u_i(\mu) \in [0,1]$ to denote the expected utility of player $i$ under $\mu$.

In a coarse correlated equilibrium, no player can gain by switching to another fixed strategy\footnote{This definition is also known as  normal-form coarse correlated equilibrium (NFCCE) in the literature}.
\begin{definition}[Coarse correlated equilibrium]
Let $\eps > 0$, a distribution $\mu \in \Delta(\mS)$ is an $\eps$-approximate coarse correlated equilibrium ($\eps$-CCE) of an extensive-form game, if for all player $i \in [m]$ and strategy $s_{i} \in \mS_i$, 
\[
u_{i}(\mu)  \geq u_{i}(s_{i}, \mu_{-i}) - \eps.
\]
\end{definition}

\begin{remark}[Mixed strategy, behaviour strategy and rank]
In the literature of Bayesian and extensive-form game, a mixed strategy places a distribution over the set of pure strategies, while a behaviour strategy places a distribution over the action set at each information set. 
The Kuhn's Theorem \cite{kuhn11953extensive}, a fundamental result of imperfect information game, states that the behaviour strategy is equivalent to the mixed strategy, in the sense that for any mixed strategy $p_i \in \Delta(\mS_i)$, one can find a behaviour strategy $x_i \in \prod_{\mI_i}\Delta(A_i)$, such that the probability of reaching any terminal node is the same.
This means that:
\begin{itemize}
    \item {\bf CCE <-> behaviour strategies:} When considering the rank of a coarse correlated equilibrium, it suffices to define the rank with respect to behaviour strategy -- one can always replace a mixed strategy with a behaviour strategy and the utility remains the same\footnote{This notion of rank/sparsity is used in the previous work of~\cite{foster2023hardness, anagnostides2024complexity}.}.
    \item {\bf CE <-> mixed strategies:} Meanwhile, when considering decentralized learning of correlated equilibrium, we should still stick with mixed strategies. The reason is that, after replacing a mixed strategy with a behaviour strategy, even though the utility received by each agent is the same, the hypothetical utility after swapping the strategy is not the same! See Appendix \ref{sec:example-app} for an example. 
\end{itemize}
\end{remark}

Formally, we define the rank of a joint distribution as follows.
\begin{definition}[Rank]
Let $\mu \in \Delta(\mS)$ be a joint distribution. When considering CE, we say $\mu$ is of rank $T$,  if it can be written as $\mu = \frac{1}{T}\sum_{t=1}^{T} \bigotimes_{i \in [m]} p_i^{(t)}$ where $p_i^{(t)} \in \Delta(S_i)$ is a mixed strategy; When considering CCE, we say $\mu$ is of rank $T$ if it can be written as 
$\mu = \frac{1}{T}\sum_{t=1}^{T} \bigotimes_{i \in [m]} x_i^{(t)}$,
where $x_{i}^{(t)} \in X_i^{\EFG}$ is a behaviour strategy.
\end{definition}

Note we define the rank of a distribution as the number of {\em uniform mixture} of product distributions. This is wlog -- if the distribution can be written as $k$ non-uniform mixture of production distributions, then one can approximate it with $O(k)$ uniform mixture.

\subsection{Useful tools}
\paragraph{No-regret learning}  We make use of the classic algorithm of Multiplicative Weights Update (MWU) for regret minimization \cite{littlestone1994weighted}. 
\begin{algorithm}[!htbp]
\caption{MWU}
\label{algo:mwu}
\begin{algorithmic}[1] 
\State {\bf Input parameters} $T$ (number of rounds),  $n$ (number of actions), $B$ (bound on payoff) 
\For{$t=1,2, \ldots, T$}
\State Compute $p_t \in \Delta_n$ over experts such that $p_t(i) \propto \exp(\eta \sum_{\tau=1}^{t-1}r_\tau(i))$ for $i\in [n]$
\State Play $p_t$ and observes $r_t \in [0, B]^n$
\EndFor
\end{algorithmic}
\end{algorithm}

\begin{lemma}[\cite{arora2012multiplicative}]
\label{lem:mwu}
Let $n, T \geq 1$ and the reward $r_t \in [0,B]^n$ ($t\in [T]$). 
If one takes $\eta = \sqrt{\log (n)/T}/B$, then the MWU algorithm guarantees an external regret of at most
\begin{align*}
\max_{i^{*}\in [n]}\sum_{t \in [T]}r_t(i^{*}) - \sum_{t \in [T]} \langle p_t, r_t\rangle \leq \frac{\log(n)}{\eta} + \eta T B^2 \leq 2B\sqrt{T\log (n)}.
\end{align*}
\end{lemma}

\paragraph{Online density estimation} We make use of the online density estimation algorithm \cite{vovk1990aggregating}. In the task of online density estimation, there is a set of outcomes $O$ and a set of possible contexts $C$. The task proceeds in $H$ rounds, and at each round $h \in [H]$,
\begin{itemize}
\item Nature first reveals a context $c_h \in C$
\item The learner predicts a distribution $p_h \in \Delta(O)$ over outcomes based on the context $c_h$; and
\item Nature chooses an outcome $o_h \in O$
\end{itemize}
The learner has access to a set of {\em experts} $\{p^{(i)}\}_{i \in [n]}$, where each expert is a fixed function from context to distribution of outcome. That is, for any context $c \in C$, $p^{(i)}_{c}\in \Delta(O)$ is a distribution of outcome, and $p^{(i)}_{c}(o)$ is the probability of outcome $o \in O$. Vovk's algorithm (see Algorithm \ref{algo:vovk}) places a distribution over experts at each round, by running MWU over the logarithmic loss function.
\begin{algorithm}
\caption{Vovk's aggregating algorithm}
\label{algo:vovk}
\begin{algorithmic}[1]
\For{$h=1,2,\ldots, H$}
\State Observe the context $c_h$
\State Compute a distribution $q_{h} \in \Delta_n$ over experts $[n]$ 
\[
q_{h}(i) = \frac{\exp\left(-\sum_{\nu=1}^{h-1}\log\left(\frac{1}{p^{(i)}_{c_\nu}(o_\nu)}\right)\right)}{\sum_{i'\in [n]}\exp\left(-\sum_{\nu=1}^{h-1}\log\left(\frac{1}{p^{(i')}_{c_\nu}(o_{\nu})}\right)\right)} \quad \forall i\in [n]
\]
\State Aggregate predictions from experts $p_{h} = \E_{i \sim q_h}[p^{(i)}_{c_h}] \in \Delta(O)$ and observe the outcome $o_h \in O$
\EndFor
\end{algorithmic}
\end{algorithm}

\begin{lemma}[\cite{vovk1990aggregating}] 
\label{lem:vovk}
Suppose the distribution of outcomes is realizable under the set of expert $[n]$, i.e., there exists $i^{*}\in [n]$ such that $o_h \sim p^{(i^{*})}_{c_h}$ for each $h \in [H]$. 
Then Vovk's aggregating rule (Algorithm \ref{algo:vovk}) satisfies
\begin{align*}
\frac{1}{H} \sum_{h=1}^{H} \E[ \TV(p_h; p^{(i^{*})}_{c_h})] \leq \sqrt{\frac{\log(n)}{H}}.
\end{align*}
\end{lemma}

\paragraph{The complexity of Bayesian Nash equilibrium}
We make use of the following hardness result of BNE.
\begin{theorem}[\cite{rubinstein2018inapproximability}]
\label{thm:lb-bayesian}
It is PPAD-hard to find an (every-type-$\eps$)-BNE of a two-player, $n$-type, $2$-action Bayesian game for some absolute constant $\eps > 0$.
\end{theorem}

\section{Lower bound for extensive-form game}
\label{sec:lower}

Theorem \ref{thm:lb-extensive-main} is a direct corollary of the following hardness result on computing low rank CCE.

\begin{theorem}
\label{thm:lb-extensive}
Let $\eps > 0$ be an absolute constant. Assuming $\PPAD \not\subset \TIME(n^{\polylog(n)})$, there is no polynomial time algorithm that can find a rank-$2^{\log_2^{1-o(1)}(|\mI|)}$ $\eps$-CCE of a three player extensive-form game.
\end{theorem}

\subsection{Construction}

We reduce (every-type-$\eps$)-BNE to low rank CCE of extensive-form game. 
Given a two-player Bayesian game $\mathcal{G}$, with type space $\Theta = \Theta_1 \times \Theta_2$, action space $A = A_1 \times A_2$ and prior distribution $\rho$, we construct a three-player extensive-form game $\Gamma$.

\paragraph{Gadget}
We first construct a three-player extensive-form game gadget that incorporates the Bayesian game.
Among the first two players, we use $i \in [2]$ to represent a player and we use $i'$ to denote the other player than $i$, i.e., $i' = 3-i \in [2]$.
We use the special symbol $\K = 3$ to represent the third player Kibitzer. 

The extensive-form game proceeds as follows. Kibitzer first picks a player $i \in [2]$, a type $\theta_i \in \Theta_{i}$ as well as an action $\wt{a}_i \in A_i$, that is, the action set $A_{\K}$ is defined as
\[
A_{\K} = \{(i, \theta_i, \wt{a}_i), i\in [2], \theta_i \in \Theta_i, \wt{a}_i \in A_i\}.
\]
Nature then samples $\theta_{i'}$  from $\rho$ conditioning on $\theta_i$. Then player $i$ (resp. $i'$) observes its private type $\theta_i$ (resp. $\theta_{i'}$) and plays action $a_{i}$ (resp. $ a_{i'}$). We use $O = A_1\times A_2\times A_{\K}\times \Theta_1\times \Theta_2$ to denote the outcome space of the gadget.

When Kibitzer plays $a_{\K} = (1,\theta_1, \wt{a}_1)$, then the utility of player $j \in [3]$ equals
\begin{align}
u_j(a_1, a_2, a_\K, \theta_1, \theta_2) = \left\{
\begin{matrix}
u_{1}(\theta_1, \theta_2; a_{1}, a_{2}) - u_{1}(\theta_1, \theta_2; \wt{a}_{1}, a_{2}) & j = 1\\
0 & j = 2\\
u_{1}(\theta_1, \theta_2; \wt{a}_{1}, a_{2}) - u_{1}(\theta_1, \theta_2; a_{1}, a_{2}) & j = K\\
\end{matrix}
\right.. \label{eq:utility1}
\end{align}
When Kibitzer plays $a_{\K} = (2,\theta_2, \wt{a}_2)$, the utility of player $j \in [3]$ is defined similarly
\begin{align}
u_j(a_1, a_2, a_\K, \theta_1, \theta_2) = \left\{
\begin{matrix}
0 & j = 1\\
u_{2}(\theta_1, \theta_2; a_{1}, a_{2}) - u_{2}(\theta_1, \theta_2; a_{1}, \wt{a}_{2}) & j = 2\\
u_{2}(\theta_1, \theta_2; a_{1}, \wt{a}_{2}) - u_{2}(\theta_1, \theta_2; a_{1}, a_{2}) & j = K\\
\end{matrix}
\right.. \label{eq:utility2}
\end{align}

\paragraph{Final construction} Let $T$ be the rank of the CCE. The final extensive-form game $\Gamma$ is constructed by repeating the gadget for $H \geq \frac{\log(T)}{\eps^2}$ times. For each repetition $h \in [H]$, let $N_h$ be the collection of gadgets at the $h$-th repetition, corresponding to different paths of play in earlier repetition. 
For any gadget $n_h \in N_h$, let $o_{h} = (a_{1,h}, a_{2,h}, a_{\K, h}, \theta_{1,h}, \theta_{2,h}) \in O$ be the outcome of the $h$-th repetition.
The utility of each of the three players is determined by the average over the gadget utilities defined Eq.~\eqref{eq:utility1}\eqref{eq:utility2}, averaged across the $H$ repetitions; we denote them as $(u_{j, n_h})_{j\in [3]}$.

\paragraph{Notation} 
A gadget $n_h \in N_h$ is uniquely determined by its history $\sigma(n_h) = (n_1, o_{1}, \ldots, n_{h-1}, o_{h-1})$, where $n_\tau \in N_{\tau}$ ($\tau \in [h-1]$) is the gadget at the root path to $n_h$ and $o_{\tau}$ is the outcome at $n_{\tau}$.
We use $\mS_{j}$ to denote the strategy space of player $j \in [3]$ in the extensive-form game $\Gamma$, and $S_i$ to denote the strategy space of player $i \in [2]$ in the Bayesian game $\mathcal{G}$. 
For player $i \in [2]$, we have $\mS_i = \prod_{h\in [H]}\prod_{n_h \in N_h} S_i$ and $X_i^{\EFG} = \prod_{h\in [H]}\prod_{n_h \in N_h} X_i$, i.e., the pure/behaviour strategy space is the product of pure/behaviour strategy space at each gadget.
Given a behaviour strategy $x_i \in X_i^{\EFG}$, for any gadget $n_h \in N_h$, we write $x_{i, n_h} \in X_i$ to denote the behaviour strategy of $x_i$ on $n_h$, i.e., $x_{i, n_h}$ is the marginal distribution of $x_{i}$ on $n_h$.
For player $\K$, we have $\mS_{\K} = \prod_{h\in [H]}\prod_{n_h \in N_h} A_{\K}$, $X_\K^{\EFG} = \prod_{h\in [H]}\prod_{n_h \in N_h} \Delta(A_\K)$, and we define $x_{\K, n_h} \in \Delta(A_{\K})$ similarly. 

Given a rank-$T$ $\eps$-CCE $\mu$ of the extensive-form game $\Gamma$, we write 
\begin{align*}
\mu = \sum_{t\in [T]} x^{(t)} = \sum_{t\in [T]} x_{1}^{(t)} \otimes x_{2}^{(t)} \otimes x_{\K}^{(t)},
\end{align*}
where $x_{j}^{(t)} \in X_j^{\EFG}$ ($j \in [3]$) is a behaviour strategy of player $j$.

\subsection{Analysis of players 1,2}
We first prove that at any approximate CCE, the  utilities of the first two players are not far below zero. To this end, for player $i \in [2]$, we consider the deviation strategy $x_{i}^{\dagger} \in \mS_{i}$, whose construction is shown at Algorithm \ref{algo:deviation-two}. For any product distribution $t \in [T]$,
player $i \in [2]$, let $x_{-i}^{(t)} = x^{(t)}_{i'}\times x^{(t)}_{\K}$. 
Let $O_{-i} = A_{i'}\times A_{\K}\times \Theta_1\times \Theta_2$ be the outcome space of player $\K$, player $i'$ and the chance player.
For any gadget $n_h \in N_h$, let $x_{-i, n_h}^{(t)} \in \Delta(O_{-i})$ be the distribution of outcome (except player $i$) at gadget $n_h$ under strategy $x_{-i}^{(t)}$.

\begin{algorithm}[!htbp]
\caption{Deviation strategy of player $1,2$ \Comment{Analysis only}}
\label{algo:deviation-two}
\begin{algorithmic}[1]
\State {\bf Input}: rank-$T$ $\eps$-CCE $\mu = \frac{1}{T}\sum_{t\in [T]} x_1^{(t)}\otimes x_{2}^{(t)} \otimes x_{\K}^{(t)}$, $i \in [2]$
\For{$h = 1,2,\ldots, H$} \Comment{repetition $h \in [H]$}
\For{$n_h \in N_{h}$} \Comment{gadget $n_h$}
\State $\sigma(n_h) = (n_1, o_1, \ldots, n_{h-1}, o_{h-1})$ \Comment{history of $n_h$}
\State Compute a distribution $q_{i, n_h} \in \Delta(T)$ over $[T]$ \Comment{Vovk's rule}
\[
q_{-i, n_h}(t) = \frac{\exp\left(-\sum_{\nu=1}^{h-1}\log\left(\frac{1}{x_{-i,n_v}^{(t)}(o_{-i, \nu})}\right)\right)}{\sum_{t'\in [T]}\exp\left(-\sum_{\nu=1}^{h-1}\log\left(\frac{1}{x_{-i, n_v}^{(t')}(o_{-i,\nu})}\right)\right)} \quad t \in [T] 
\]  
\State Compute $p_{-i, n_h} = \E_{t\sim q_{-i, n_h}}[x_{-i, n_{h}}^{(t)}] \in \Delta(O_{-i})$  \Comment{predicted outcome}
\State Define the deviation strategy $x^{\dagger}_{i, n_h} \in S_i$ at gadget $n_h$
\begin{align*}
x_{i, n_h}^{\dagger} = &~ \argmax_{s_i\in S_{i}} \E_{(a_{i'},a_{\K}, \theta_1, \theta_2) \sim p_{-i, n_h}}[u_{i, n_h}(s_i(\theta_i), a_{i'}, a_{\K}, \theta_1,\theta_2)]
\end{align*}
\EndFor
\EndFor
\end{algorithmic}
\end{algorithm}

First, we prove
\begin{lemma}
\label{lem:two-tv}
For player $i\in [2]$, product distribution $t \in [T]$, we have
\begin{align*}
\frac{1}{H} \sum_{h\in [H]} \E_{x_i^{\dagger} \times x_{-i}^{(t)}} \left[ \TV(p_{-i, n_h}; x_{-i, n_h}^{(t)}) \right] \leq \sqrt{\frac{\log(T)}{H}}.
\end{align*}
Here $n_1, n_2, \ldots, n_H$ is the (random) sequence of gadgets visited under strategy $x_i^{\dagger} \times x_{-i}^{(t)}$.
\end{lemma}
\begin{proof}
Consider the following online density estimation procedure. The outcome space is $O_{-i}$, the set of experts are $T$ strategies $\{x_{-i, n_h}^{(t)}\}_{t\in [T]}$, the context comes from the set of gadgets $\cup_{h\in [H]}N_h$. At each round $h \in [H]$, 
\begin{itemize}
\item Nature provides the context $n_h$
\item The algorithm predicts the distribution $p_{-i, n_h} \in \Delta(O_{-i})$ over outcome space $O_{-i}$;  
\item Nature reveals the outcome $o_{-i, h} = (a_{i'}, a_{\K}, \theta_1, \theta_2) \sim x_{-i, n_h}^{(t)}$. 
\item The next context $n_{h+1} = n_h | (o_{-i,h}, x_{i}^{\dagger}(\theta_i))$
\end{itemize}
By Lemma \ref{lem:vovk}, we have
\begin{align*}
\frac{1}{H}\sum_{h=1}^{H}\E_{x_i^{\dagger} \times x_{-i}^{(t)}} \left[  \TV(p_{-i, n_h}; x_{-i, n_h}^{(t)}) \right]  \leq \sqrt{\frac{\log(T)}{H}}.
\end{align*}
\end{proof}

 Now we can bound the utility of the first two players at $\eps$-CCE.
\begin{lemma}[Equilibrium utility] \label{lem:utility}
The utility at $\eps$-CCE of player $i \in [2]$ satisfies $u_i(\mu) \geq -5\eps$. 
\end{lemma}
\begin{proof}
For any player $i\in [2]$, production distribution $t \in [T]$, gadget $n_h \in N_h$, define 
\[
s^{(t)}_{i, n_h} = \argmax_{s_i \in S_i} \E_{(a_{i'}, a_{\K}, \theta_1,\theta_2) \sim x_{-i, n_h}^{(t)}}u_{i, n_h}(s_{i}(\theta_i), a_{i'}, a_{\K}, \theta_1, \theta_2).
\]
That is $s^{(t)}_{i, n_h} \in S_{i}$ is the best strategy of player $i$, given that the other players use strategies $x_{i', n_h}^{(t)} \times x_{\K, n_h}^{(t)}$. We claim that
\begin{align}
u_{i, n_h}(s^{(t)}_{i, n_h}, x^{(t)}_{-i, n_h}) \geq 0. \label{eq:non-negative}
\end{align}
This is because player $i$ can always best-respond to $x_{i', n_h}^{(t)}$ in the Bayesian game, and the strategy $x_{\K, n_h}^{(t)}$ of Kibitzer is independent of $x_{i'}^{(t)}$ so it cannot find a better action.

Now we compute its utility under deviation $x_{i}^{\dagger}$.  For any $t \in [T]$, gadget $n_h\in N_h$, we have
\begin{align}
&~ u_{i, n_h}(x_{i, n_h}^{\dagger}, x_{-i, n_h}^{(t)}) \notag \\
= &~ u_{i, n_h}(x_{i, n_h}^{\dagger}, x_{-i, n_h}^{(t)})- u_{i, n_h}(x_{i, n_h}^{\dagger}, p_{-i, n_h}) + u_{i, n_h}(x_{i, n_h}^{\dagger}, p_{-i, n_h}) - u_{i, n_h}(s_{i, n_h}^{(t)}, p_{-i, n_h})\notag \\
&~ + u_{i, n_h}(s_{i, n_h}^{(t)}, p_{-i, n_h}) - u_{i, n_h}(s_{i, n_h}^{(t)}, x_{-i, n_h}^{(t)}) + u_{i, n_h}(s_{i, n_h}^{(t)}, x_{-i, n_h}^{(t)}) \notag \\
\geq &~ - \frac{2}{H}\TV(p_{-i, n_h}; x_{-i, n_t}^{(t)}) + 0  - \frac{2}{H}\TV(p_{-i, n_h}; x_{-i, n_t}^{(t)}) + u_{i, n_h}(s_{i, n_h}^{(t)}, x_{-i, n_h}^{(t)})\notag \\
\geq &~ -\frac{4}{H} \TV(p_{-i, n_h}; x_{-i, n_t}^{(t)}).\label{eq:two-utility1}
\end{align}
The second step follows from (i) the value of $u_{i, n_h}$ is between $[-\frac{1}{H}, \frac{1}{H}]$ and (ii) $x_{i, n_h}^{\dagger}$ is the best response under $p_{-i, n_h}$. The third step follows from Eq.~\eqref{eq:non-negative}.

Taking an expectation, we have
\begin{align*}
u_{i}(x_i^{\dagger}, x_{-i}^{(t)}) = &~ \sum_{h=1}^{H} \E_{x_i^{\dagger}\times x_{-i}^{(t)}} \left[u_{i, n_h}(x_{i, n_h}^{\dagger}, x_{-i, n_h}^{(t)})\right] \\
\geq &~ -\frac{4}{H}\sum_{h=1}^{H} \E_{x_i^{\dagger}\times x_{-i}^{(t)}} \left[\TV(p_{-i, n_h}; x_{-i,n_h}^{(t)} ) \right]\\
\geq &~ -4\sqrt{\frac{\log(T)}{H}} \geq -4\eps.
\end{align*}
Here the second step follows from Eq.~\eqref{eq:two-utility1}, the third step follows from Lemma \ref{lem:two-tv}.

Taking an expectation over $t \sim [T]$, we have
\begin{align*}
u_{i}(x_i^{\dagger}, \mu_{-i}) = \frac{1}{T}\sum_{t\in [T]}u_{i}(x_i^{\dagger}, x_{-i}^{(t)}) \geq -4\eps.
\end{align*}
Note $\mu$ is an $\eps$-CCE, we have $u_i(\mu) \geq -5\eps$
\end{proof}

A direct corollary of Lemma \ref{lem:utility} is that the utility of player $\K$ at $\eps$-CCE is not much above $0$.
\begin{lemma}
\label{lem:eq-utility-k1}
The utility at $\eps$-CCE of player $\K$ satisfies $u_{\K}(\mu) \leq 10\eps$. 
\end{lemma}
\begin{proof}
The utility of three players sums to $0$. By Lemma \ref{lem:utility}, we have that 
\[
u_\K(\mu) = -u_{1}(\mu) - u_2(\mu) \leq 10\eps.\qedhere
\]
\end{proof}

\subsection{Reduction}
The reduction is formally presented at Algorithm \ref{algo:reduction}, using some of the following notation.
Let $O_{-\K} = A_1\times A_2 \times (\Theta_{1}\cup \Theta_2)$ be the outcome space of player $1,2$ and the chance player.
For any $t \in [T]$, let $x_{-\K}^{(t)} = x^{(t)}_{1}\times x^{(t)}_{2}$. 
For any gadget $n_h \in N_h$ and action $a_{\K} = (i, \theta_i, \wt{a}_{i}) \in A_{\K}$, let $x_{-\K, n_h, a_{\K}}^{(t)} \in \Delta(O_{-\K})$ be the distribution of outcomes (except player $\K$) at gadgets $n_h$, given that player $\K$ plays $a_{\K}$ and players $1,2$ use (mixed) strategy profile $x_{1, n_{1}}^{(t)}\times x_{2,n_2}^{(t)}$.

\begin{algorithm}[!htbp]
\caption{Reduction}
\label{algo:reduction}
\begin{algorithmic}[1]
\State {\bf Input}: rank-$T$ $\eps$-CCE $\mu = \frac{1}{T}\sum_{t\in [T]}  x_1^{(t)}\otimes x_{2}^{(t)} \otimes x_{\K}^{(t)}$ 
\For{$h = 1,2,\ldots, H$}\Comment{repetition $h$}
\For{$n_h \in N_{h}$} \Comment{gadget $n_h$}
\State $\sigma(n_h) = (n_1, o_1, \ldots, n_{h-1}, o_{h-1})$ \Comment{history of $n_h$}
\State Compute a distribution $q_{-\K, n_h} \in \Delta(T)$ over $[T]$
\[
q_{-\K, n_h}(t) = \frac{\exp\left(-\sum_{\nu=1}^{h-1}\log\left(\frac{1}{x_{-\K, n_v, a_{\K, \nu}}^{(t)}(o_{-\K, \nu})}\right)\right)}{\sum_{t'\in [T]}\exp\left(-\sum_{\nu=1}^{h-1}\log\left(\frac{1}{x_{-\K, n_v, a_{\K, \nu}}^{(t')}(o_{-\K, \nu})}\right)\right)} \quad \forall t \in [T]
\] 
\State Compute $p_{i, n_h} = \E_{t \sim q_{-\K, n_h}} [x_{i, n_h}^{(t)}] \in \Delta(S_i)$ for $i = 1,2$ \label{line:marginal}
\If{$(p_{1, n_h}, p_{2, n_h})$ is an (every-type-$16\eps$)-BNE} \label{line:BNE}
\State \Return $(p_{1, n_h}, p_{2, n_h})$
\Else
\State Choose $a_{\K, n_{h}}^{\dagger} = (i, \theta_i, \wt{a}_{i})\in A_{\K}$ such that 
\[
\E_{\theta_{i'} \sim \rho | \theta_{i}}\E_{s_i \sim p_{i, n_h}}\E_{s_{i'}\sim p_{i',n_h}}[u_{i}(\theta_{i}, \theta_{i'}; \wt{a}_i, s_{i'}(\theta_{i'})) - u_{i}(\theta_i, \theta_{i'}; s_{i}(\theta_i), s_{i'}(\theta_{i'}))] \geq 16\eps.
\]\label{line:not-well-support}
\EndIf
\EndFor
\EndFor
\end{algorithmic}
\end{algorithm}

For any node $n_h\in N_h$, $a_{\K} = (i, \theta_i, \wt{a}_{i}) \in A_{\K}$, define $p_{-\K, n_h, a_{\K}} \in \Delta(\Theta_{i'} \times A_1\times A_2)$ 
such that 
\[
p_{-\K, n_h, a_{\K}}(\theta_{i'}, a_{i}, a_{i'}) = \Pr_{\hat{\theta}_{i'} \sim \rho|\theta_i}[\hat{\theta}_{i'} = \theta_{i'}]\cdot \Pr_{s_i \sim p_{i, n_h}}[a_i = s_i(\theta_{i})]\cdot \Pr_{s_{i'}\sim p_{i', n_h}}[a_{i'} = s_{i'}(\theta_{i'})]
\]
That is, it first samples $\theta_{i'} \sim \rho|\theta_{i}$, then it independently samples strategies $s_{i}, s_{i'}$  from $p_{i, n_h}, p_{i', n_h}$ (respectively), and finally takes $a_i = s_i(\theta_i), a_{i'} = s_{i'}(\theta_{i'})$. 
For $i \in [2]$, $\theta_i\in \Theta_i$, we further define $p_{i, n_h, \theta_i} \in \Delta(A_i)$ such that $p_{i, n_h, \theta_i}(a_i) = \E_{s_i\sim p_{i, n_h}}[s_i(\theta_i)]$.

Our goal is to prove, if Algorithm~\ref{algo:reduction} doesn't find an approximate BNE, then Kibitzer has a good deviation strategy $x_{\K}^{\dagger}$.
First, we prove Kibitzer approximately learns the outcome distribution.
\begin{lemma}
\label{lem:k-tv}
Suppose Algorithm~\ref{algo:reduction} does not find an (every-type-$16\eps$)-BNE, i.e.~the if-statement in Line~\ref{line:BNE} is always false,
For every $t \in [T]$, we have
\begin{align*}
\frac{1}{H}\sum_{h\in [H]} \E_{x_\K^{\dagger} \times x_{-\K}^{(t)}} \left[  \TV(p_{-\K, n_h, a_{\K, h}^{\dagger}}; x_{-\K, n_h, a_{\K, h}^{\dagger}}^{(t)}) \right] \leq 2\sqrt{\frac{\log(T)}{H}}.
\end{align*}
Here $n_{1}, \ldots, n_{H}$ is the (random) sequence of gadgets visited under $x_{\K}^{\dagger}\times x_{-\K}^{\dagger}$ and $a_{\K, h}^{\dagger} = a_{\K, n_h}^{\dagger}$ is the action of player $\K$ at gadget $n_h$.
\end{lemma}
\begin{proof}
For any gadget $n_h \in N_h$, define 
\[
w_{-\K, n_h, a^{\dagger}_{\K, h}} = \E_{t' \sim q_{-\K, n_h}} \left[x_{-\K, n_h, a_{\K, h}^{\dagger}}^{(t')}\right] \in \Delta(O_{-\K}).
\]
We first bound the (expected) TV distance between $w_{-\K, n_h, a_{\K, h}^{\dagger}}$ and $x_{-\K, n_h, a_{\K, h}^{\dagger}}^{(t)}$.
Consider the following online density estimation procedure. The outcome space is $O_{-\K}$, the set of experts are $T$ indices $\{x_{-\K}^{(t)}\}_{t\in [T]}$ , and the context comes from $\cup_{h \in [H]}N_h\times A_\K$.
At round $h\in [H]$, 
\begin{itemize}
\item Nature provides the context $(n_h, a_{\K, h}^{\dagger})$;
\item The algorithm predicts $w_{-\K, n_h, a_{\K, h}^{\dagger}} \in \Delta(O_{-\K})$; 
\item Nature samples the outcome $o_{-\K, h} \sim x_{-\K, n_h, a_{\K, h}^{\dagger}}^{(t)}$;
\item The next context $n_{h+1} = n_h | (a_{\K,h}^{\dagger}, o_{-\K, h})$ and $a_{\K, h+1}^{\dagger} = a_{\K,n_{h+1}}^{\dagger}$
\end{itemize}
By the guarantee of Lemma \ref{lem:vovk}, we have
\begin{align}
\frac{1}{H}\sum_{h\in [H]} \E_{x_\K^{\dagger} \times x_{-\K}^{(t)}} \left[  \TV(w_{-\K, n_h, a_{\K, h}^{\dagger}}; x_{-\K, n_h, a_{\K, h}^{\dagger}}^{(t)} )\right] \leq \sqrt{\frac{\log(T)}{H}}.\label{eq:k-tv-density}
\end{align}

Now consider the three distributions $p_{-\K, n_h, a_{\K, h}^{\dagger}}, w_{-\K, n_h, a_{\K, h}^{\dagger}}, x_{-\K, n_h, a_{\K, h}^{\dagger}}^{(t)} \in \Delta(\Theta_{i'}\times A_1\times A_2)$.
Intuitively, $x_{-\K, n_h, a_{\K, h}^{\dagger}}^{(t)}$ is the outcome under the product distribution $x_{-i}^{(t)}$;  $w_{-\K, n_h, a_{\K, h}^{\dagger}}$ is the predicted outcome under Vovk's algorithm; $p_{-\K, n_h, a_{\K, h}^{\dagger}}$ is the outcome after taking marginal of $w_{-\K, n_h, a_{\K, h}^{\dagger}}$.
Formally, these three distributions have the same marginal distribution on $\Theta_{i'}$ due to the independence of chance player.  Conditioning on the first coordinates being $\theta_{i'} \in \Theta_{i'}$
\begin{itemize}
\item $x_{-\K, n_h, a_{\K, h}^{\dagger}}^{(t)} | \theta_{i'} = x_{1, n_h, \theta_1}^{(t)} \times x_{2, n_h, \theta_1}^{(t)} \in \Delta(A_1)\times \Delta(A_2)$ is a product distribution;  
\item $p_{-\K, n_h, a_{\K, h}^{\dagger}} | \theta_{i'} = p_{1, n_h, \theta_1} \times p_{2, n_h, \theta_2} \in \Delta(A_1)\times \Delta(A_2)$ is a product distribution; while
\item $w_{-\K, n_h, a_{\K, h}^{\dagger}}| \theta_{i'}  = \E_{t\sim q_{-\K, n_h}} [x_{1, n_h, \theta_1}^{(t)} \times x_{2, n_h, \theta_2}^{(t)}] \in \Delta(A_1\times A_2)$ is a joint distribution over $A_1\times A_2$
\item $p_{-K, n_h, a_{K, h}^{\dagger}} | \theta_{i'}$ and $w_{-K, n_h, a_{K, h}^{\dagger}} | \theta_{i'}$ have the same marginal distribution over $a_1 \in A_1$ (resp. $a_2 \in A_2$), because $p_{1, n_h, \theta_1} = \E_{t\sim q_{-\K, n_h}} [x_{1, n_h, \theta_1}^{(t)}]$ and $p_{2, n_h, \theta_2} = \E_{t\sim q_{-\K, n_h}} [x_{2, n_h, \theta_2}^{(t)}]$  (see Line \ref{line:marginal} of Algorithm \ref{algo:reduction})
\end{itemize}



Then we have
\begin{align*}
\TV(p_{-\K, n_h, a_{\K, h}^{\dagger}}; x_{-\K, n_h, a^{\dagger}_{\K, h}}^{(t)}) = &~  \E_{\theta_{i'} \sim \rho | \theta_i} \left[\TV\left(p_{1, n_h, \theta_1} \times p_{2, n_h, \theta_2}\;\; ; \;\; x_{1, n_h, \theta_1}^{(t)}\times x_{2, n_h, \theta_2}^{(t)}\right)\right]\\
\leq &~ 2 \E_{\theta_{i'} \sim \rho|\theta_i} \left[\TV\left(w_{-\K,n_h, a_{\K, h}^{\dagger}}|\theta_{i'}\;\; ; \;\; x_{1, n_h, \theta_1}^{(t)}\times x_{2, n_h, \theta_2}^{(t)}\right)\right]\\
= &~ 2 \TV\left(w_{-\K,n_h, a_{\K, h}^{\dagger}}\;\; ; \;\; x_{-\K, n_h, a_{\K, h}^{\dagger}}^{(t)}\right).
\end{align*}
Here the first step follows from Fact \ref{fact:product1}, the second follows from Fact \ref{fact:product2} and the last step follows from Fact \ref{fact:product1}.

Taking an expectation and combining with Eq.~\eqref{eq:k-tv-density}, we have
\begin{align*}
&~\frac{1}{H}\sum_{h\in [H]} \E_{ x_\K^{\dagger} \times x_{-\K}^{(t)}} \left[  \TV(p_{-\K, n_h, a_{\K, h}^{\dagger}}; x_{-\K, n_h, a_{\K, h}^{\dagger}}^{(t)}) \right]\\ \leq &~ \frac{2}{H}\sum_{h\in [H]} \E_{x_\K^{\dagger} \times x_{-\K}^{(t)}} \left[  \TV(w_{-\K, n_h, a_{\K, h}^{\dagger}}; x_{-\K, n_h, a_{\K, h}^{\dagger}}^{(t)}) \right] \leq 2\sqrt{\frac{\log(T)}{H}}.
\end{align*}
\end{proof}

Now we can prove the correctness of the reduction.
\begin{lemma}
\label{lem:reduction}
Algorithm \ref{algo:reduction} returns an (every-type-$16\eps$)-BNE.
\end{lemma}
\begin{proof}
We prove by contradiction and assume Algorithm \ref{algo:reduction} never returns. We compute the utility under deviation $x_{\K}^{\dagger}$. For any product distribution $t \in [T]$, we have
\begin{align*}
u_{\K}(x_\K^{\dagger}, x_{-\K}^{(t)}) = &~ \sum_{h=1}^{H} \E_{x_\K^{\dagger}\times x_{-\K}^{(t)}} \left[u_{\K, n_h}(a_{\K, h}^\dagger, x_{-\K, n_h, a_{\K, h}^{\dagger}}^{(t)})\right] \\
\geq &~ \sum_{h=1}^{H} \E_{x_\K^{\dagger}\times x_{-\K}^{(t)}} \left[-\frac{2}{H} \TV(p_{-\K, n_h, a_{\K, h}^{\dagger}}; x_{-\K, n_h, a_{\K, h}^{\dagger}}^{(t)}) + u_{\K, n_h}(a_{\K, h}^\dagger, p_{-\K, n_h, a_{\K, h}^{\dagger}}^{(t)})\right] \\
\geq &~ \sum_{h=1}^{H} \E_{x_\K^{\dagger}\times x_{-\K}^{(t)}} \left[-\frac{2}{H} \TV(p_{-\K, n_h, a_{\K, h}^{\dagger}}; x_{-\K, n_h, a_{\K, h}^{\dagger}}^{(t)}) + \frac{16\eps}{H}\right] \\
\geq &~ -4\sqrt{\frac{\log(T)}{H}} + 16\eps \geq 12\eps. 
\end{align*}
The second step holds since $u_{K, n_h}$ is between $[-\frac{1}{H}, \frac{1}{H}]$, the third step follows from the choice $a_{\K, h}^{\dagger}$ (see Line \ref{line:not-well-support} of Algorithm \ref{algo:reduction}). The fourth step follows from Lemma \ref{lem:k-tv}. 

Taking an expectation over $t \sim [T]$, we have
\begin{align*}
u_{\K}(x_\K^{\dagger}, \mu_{-\K}) = \frac{1}{T}\sum_{t\in [T]}u_{i}(x_\K^{\dagger}, x_{-\K}^{(t)}) \geq 12\eps.
\end{align*}
Note $\mu$ is an $\eps$-CCE, we have $u_\K(\mu) \geq u_{\K}(x_\K^{\dagger}, \mu_{-\K}) -\eps \geq 11\eps$. This contradicts Lemma \ref{lem:eq-utility-k1}, hence, we have proved Algorithm \ref{algo:reduction} must return an (every-type-$16\eps$)-BNE at some step.
\end{proof}


Now we can finish the proof of Theorem \ref{thm:lb-extensive}.
\begin{proof}[Proof of Theorem \ref{thm:lb-extensive}]
Let $\delta > 0$ be the absolute constant in Theorem \ref{thm:lb-bayesian}.
We prove by contradiction and assume that there exists a constant $\alpha > 0$, such that one can find an rank-$2^{\log_2^{1 - \alpha}(|\mI|)}$ $\frac{\delta}{16}$-CCE of a three-player extensive-form game with runtime polynomial in $|\mI|$.
Given a two-player, $n$ types, two-action Bayesian game $\mathcal{G}$, we construct the extensive-form game $\Gamma$ with $H = (\frac{\log_2(n)}{\delta^2})^{1/\alpha}$ repetitions. 
The number of information sets at each gadget is $O(n^2)$, and therefore, the total number of information sets $|\mI| = n^{O(H)}$. Moreover, the number of repetition $H$ satisfies
\[
\frac{\log(2^{\log_2^{1 - \alpha}(|\mI|)})}{(\delta/16)^2} = O(\delta^{-2} \log_2^{1-\alpha}(|\mI|)) =O(H^{1-\alpha}\log_2^{1-\alpha}(n)\delta^{-2}) \leq H.
\]
Hence, by Lemma \ref{lem:reduction}, one can find an (every-type-$\delta$)-BNE in time $
\poly(|\mI|) = n^{O(H)} = n^{\polylog(n)}$, which, by Theorem~\ref{thm:lb-bayesian} contradicts the assumption that $\PPAD \not\subset \TIME(n^{\polylog(n)})$.  
\end{proof}

 \section{Efficient algorithm for Bayesian game}
\label{sec:bayesian}

Our goal is to prove 
\begin{theorem}[Formal statement of Theorem \ref{thm:bayesian-regret-main}]
\label{thm:bayesian-regret}
Let $\eps > 0$, for any $m$-player, $n$-action, $K$-type Bayesian game, there exists uncoupled dynamics that converge to the set of (every-type-$\eps$)-NFCE in $(\log(n)/\eps^2)^{O(1/\eps)}$ iterations. The computation cost at each iteration is $O(mnK\log(mnK/\eps)/\eps^3)$.
\end{theorem}

We have each player run the no-regret learning algorithm (shown in Algorithm \ref{algo:bayesian}) for a sequence of $T$ days. Algorithm \ref{algo:bayesian} is an adaptation of the multi-scale MWU \cite{peng2023fast} for Bayesian game.
Given the regret parameter $\eps > 0$, it sets $H = \log(n)/\eps^2$, $L = 1/\eps$ and runs for $T = H^{L}$ days. It maintains $L$ threads and plays uniformly over them.
The $L$ threads restart and update in different frequency. 
The key innovation of Algorithm \ref{algo:bayesian} is that it runs MWU for each type $k \in [K]$ independently (see Procedure $\textsc{MWU}_{\ell, k}$) and the strategy of the $\ell$-th thread is the product strategy of each type (see Line \ref{line:product}). Concretely, for each type $k \in [K]$, $\textsc{MWU}_{\ell, k}$ runs MWU and the reward vector is constructed as the expected reward conditioning on type $k$ (see Line \ref{line:reward}).

\begin{algorithm}[!htbp]
\caption{\underline{Multi-scale MWU for Bayesian game} \\
\Comment{No-regret learning algorithm for player $i \in [m]$, the opponents play $\mu_{t}^{(-i)}\in \Delta(S_{-i})$ at day $t \in [T]$ }}
\label{algo:bayesian}
\begin{algorithmic}[1]
\State {\bf Parameters} $H = \log(n)/\eps^2$, $L = 1/\eps$, $T = H^{L}$ 
\State Let $q_{t, \ell} \in (\Delta(A_i))^{\Theta_i}$ be the strategy of $\textsc{Thread}_\ell$ ($\ell \in [L]$) at day $t\in [T]$
\State Play the uniform mixture $p_{t} = \frac{1}{L} \sum_{\ell \in [L]}q_{t, \ell} \in \Delta(S_i)$\\

\Procedure{$\textsc{Thread}_\ell$}{} \Comment{Thread $\ell \in [L]$}
\State Let $w_{t, \ell, k} \in \Delta(A_i)$  be the strategy of $\textsc{MWU}_{\ell, k}$ ($k\in [K]$) at day $t\in [T]$
\State Play the product distribution $q_{t, \ell} = w_{t, \ell, 1} \times \cdots \times w_{t, \ell, K} \in (\Delta(A_i))^{\Theta_i}$ \label{line:product}
\EndProcedure\\

\Procedure{$\textsc{MWU}_{\ell, k}$}{} \Comment{Thread $\ell \in [L]$, type $k\in [K]$}
\For{$\beta = 1,2,\ldots, T/H^{\ell}$} \Comment{$\beta$-th restart, restart every $H^{\ell}$ days}
\State Initiate MWU with parameters $H, n, H^{\ell-1}$ \label{line:meta-day1}
\For{$h = 1,2,\ldots, H$}\Comment{Round $h$}
\State Let $z_{\beta, h, \ell, k}\in\Delta(A_i)$ be the strategy of MWU at the $h$-th round
\State Play $w_{t,\ell, k} = z_{\beta, h, \ell, j}$ at day $t \in [(\beta-1) H^{\ell}+ (h-1) H^{\ell-1} + 1: (\beta-1) H^{\ell}+ h H^{\ell-1}]$ 
\State Update MWU with the aggregated rewards of the past $H^{\ell-1}$ days 
\[
\left\{\sum_{t =(\beta-1) H^{\ell}+ (h-1) H^{\ell-1} + 1}^{(\beta-1) H^{\ell}+ h H^{\ell-1}} 
r_{t, k}(j)\right\}_{j \in [n]} \in [0, H^{\ell-1}]^n
\]
\State where 
\[
r_{t, k}(j) = \E_{s_{-i} \sim \mu_{t}^{(-i)}}\E_{\theta_{-i} \sim (\rho | \theta_i = k)}[u_i(\theta; j, s_{-i}(\theta_{-i}))] \quad \forall j\in [n]
\]\Comment{$\mu_{t}^{(-i)} \in \Delta(S_{-i})$} \label{line:reward}
\EndFor 
\EndFor
\EndProcedure
\end{algorithmic}
\end{algorithm}

We first bound the external regret for each thread $\ell \in [L]$ and each type $k\in [K]$.
\begin{lemma}[External regret]
\label{lem:external}
For each thread $\ell \in [L]$,  restart $\beta \in [T/H^{\ell}]$, type $k \in [K]$ and action $j \in [n]$, we have
\begin{align*}
\sum_{t = (\beta-1)H^{\ell}+1}^{\beta H^{\ell}}r_{t, k}(j) - \sum_{t = (\beta-1)H^{\ell}+1}^{\beta H^{\ell}}\sum_{s_i \in S_i}q_{t, \ell}(s_i) r_{t, k}(s_i(k))   \leq 2\eps H^{\ell}.
\end{align*}
\end{lemma}
\begin{proof}
We have 
\begin{align*}
\sum_{t = (\beta-1)H^{\ell}+1}^{\beta H^{\ell}}\sum_{s_i \in S_i}q_{t, \ell}(s_i) r_{t, k}(s_i(k)) = &~ \sum_{t = (\beta-1)H^{\ell}+1}^{\beta H^{\ell}}\sum_{j\in [n]} \sum_{s_i \in S_i, s_i(k) = j}q_{t, \ell}(s_i) r_{t, k}(j)\\
= &~\sum_{t = (\beta-1)H^{\ell}+1}^{\beta H^{\ell}}\sum_{j\in [n]} w_{t, \ell, k}(j) r_{t, k}(j)\\
\geq &~\max_{j \in [n]}\sum_{t = (\beta-1)H^{\ell}+1}^{\beta H^{\ell}}r_{t, k}(j) - 2 \sqrt{H \log(n)} H^{\ell-1}\\
= &~\max_{j \in [n]}\sum_{t = (\beta-1)H^{\ell}+1}^{\beta H^{\ell}}r_{t, k}(j) - 2 \eps H^{\ell}.
\end{align*}
Here the second step holds since the strategy $q_{t, \ell} = w_{t, \ell, 1}\times \cdots \times w_{t, \ell, K}$ is the product distribution of $\{w_{t,\ell,k}\}_{k\in [K]}$. The third step holds due to the regret guarantee of MWU, the last step holds due to $H =\log(n)/\eps^2$. 
\end{proof}

The next step is to bound the swap regret from the external regret guarantee. 
\begin{lemma}[External to swap regret reduction]
\label{lem:external-swap}
For any type $k \in [K]$ and any swap function $\phi: S_i\rightarrow S_i$, we have
\begin{align*}
\sum_{t\in [T]}\sum_{s_i \in S_i}p_{t}(s_i) r_{t, k}(\phi(s_i)(k)) - p_{t}(s_i) r_{t, k}(s_i(k))   \leq 3\eps T.
\end{align*}
\end{lemma}
\begin{proof}
First, we have
\begin{align}
&~ \sum_{t\in [T]}\sum_{s_i \in S_i}p_{t}(s_i) r_{t, k}(\phi(s_i)(k)) - p_{t}(s_i) r_{t, k}(s_i(k)) \notag \\
=&~ \frac{1}{L}\sum_{\ell\in [L]}\sum_{t \in [T]}\sum_{s_i \in S_i}q_{t, \ell}(s_i) r_{t, k}(\phi(s_i)(k))  - q_{t, \ell}(s_i) r_{t, k}(s_i(k)) \label{eq:regret-reduction1}
\end{align}
since $p_{t} = \frac{1}{L}\sum_{\ell\in [L]}q_{t, \ell}$.

For the first term, for each thread $\ell \in [L]$, we have
\begin{align}
\sum_{t \in [T]}\sum_{s_i \in S_i}q_{t, \ell}(s_i) r_{t, k}(\phi(s_i)(k)) = &~ \sum_{\beta \in [T/H^{\ell}]}\sum_{h\in [H]} \sum_{t =(\beta-1) H^{\ell}+ (h-1) H^{\ell-1} + 1}^{(\beta-1) H^{\ell}+ h H^{\ell-1}} \sum_{s_i \in S_i}q_{t, \ell}(s_i) r_{t, k}(\phi(s_i)(k)) \notag \\
\leq &~ \sum_{\beta \in [T/H^{\ell}]}\sum_{h\in[H]} \left\|\sum_{t =(\beta-1) H^{\ell}+ (h-1) H^{\ell-1} + 1}^{(\beta-1) H^{\ell}+ h H^{\ell-1}} r_{t, k}\right\|_{\infty}\label{eq:regret-reduction2}
\end{align}
where the second step holds since $q_{t, \ell} = w_{t, \ell,1}\times \cdots \times w_{t,\ell, K}$ is fixed during the interval $t \in [(\beta-1) H^{\ell}+ (h-1) H^{\ell-1} + 1: (\beta-1) H^{\ell}+ h H^{\ell-1}]$.

For the second term in Eq.~\eqref{eq:regret-reduction1}, for each thread $\ell \in [L]$ we have
\begin{align}
 \sum_{t \in [T]}\sum_{s_i \in S_i}q_{t, \ell}(s_i) r_{t, k}(s_i(k)) = &~ \sum_{\beta \in [T/H^{\ell}]} \sum_{t =(\beta-1)H^{\ell}+  1}^{\beta H^{\ell}} \sum_{s_i \in S_i}q_{t, \ell}(s_i) r_{t, k}(s_i(k))\notag \\
\geq &~ \sum_{\beta \in [T/H^{\ell}]} \left(\left\|\sum_{t =(\beta-1) H^{\ell}+  1}^{\beta H^{\ell}} r_{t, k}\right\|_{\infty} - 2\eps H^{\ell}\right)\notag\\
= &~  \sum_{\beta \in [T/H^{\ell}]} \left\|{\sum_{t =(\beta-1) H^{\ell}+  1}^{\beta H^{\ell}}} r_{t, k}\right\|_{\infty} - 2\eps T\label{eq:regret-reduction3}
\end{align}
where the second step follows from Lemma \ref{lem:external}.

Combining Eq.~\eqref{eq:regret-reduction1}\eqref{eq:regret-reduction2}\eqref{eq:regret-reduction3}, we have
\begin{align*}
\sum_{t\in [T]}\sum_{s_i \in S_i}p_{t}(s_i) r_{t, k}(\phi(s_i)(k)) - p_{t}(s_i) r_{t, k}(s_i(k)) \leq \frac{1}{L}\sum_{t\in [T]}\|r_{t, k}\|_{\infty}+ 2\eps T \leq 3\eps T
\end{align*}
Here the second step follows from $L = 1/\eps$ and $\|r_{t, k}\|_{\infty} \leq 1$.
\end{proof}

Now we can prove Theorem \ref{thm:bayesian-regret}
\begin{proof}[Proof of Theorem \ref{thm:bayesian-regret}]
We make each player $i \in [m]$ run Algorithm \ref{algo:bayesian} for $T = (\log(n)/\eps^2)^{1/\eps}$ rounds.
Let $p^{(i)}_t \in \Delta(S_i)$ be the strategy distribution of player $i$ at round $t \in [T]$, then the joint distribution is $\mu_t = \bigotimes_{i\in [m]} p_t^{(i)}$. We take $\mu_{t}^{(-i)} = \bigotimes_{i' \neq i} p_t^{(i')}$ for any $i \in [m]$.
We prove the empirical distribution $\mu = \frac{1}{T}\sum_{t\in [T]}\mu_t$ is a (every-type-$3\eps$)-NFCE. 

For each player $i \in [m]$, type $k \in [K]$, swap function $\phi: S_i \rightarrow S_i$, we have
\begin{align*}
&~ \E_{\theta\sim (\rho | \theta_i = k )}  \E_{s\sim \mu}[u_i(\theta; \phi(s_i)(\theta_i), s_{-i}(\theta_{-i}))- u_i(\theta; s(\theta))] \\
= &~ \frac{1}{T}\E_{\theta\sim (\rho | \theta_i = k )} \sum_{t\in [T]}\sum_{s_i \in S_i}p_t^{(i)}(s_i) \E_{s_{-i} \sim \mu_{t}^{(-i)}} [u_i(\theta; \phi(s_i)(k), s_{-i}(\theta_{-i})) - u_i(\theta; s_i(k), s_{-i}(\theta_{-i}))] \\
= &~ \frac{1}{T}\sum_{t\in [T]}\sum_{s_i \in S_i}p_{t}^{(i)} (s_i) \underbrace{\E_{s_{-i} \sim \mu_{t}^{(-i)}} \E_{\theta\sim (\rho | \theta_i = k )} [ u_i(\theta; \phi(s_i)(k), s_{-i}(\theta_{-i})) - u_i(\theta; s_i(k), s_{-i}(\theta_{-i}))]}_{r_{t, k}^{(i)}(\phi(s_i)(k)) - r_{t, k}^{(i)}(s_i(k))}\\
= &~ \frac{1}{T} \sum_{t\in [T]}\sum_{s_i \in S_i}p_{t}^{(i)} (s_i) (r_{t, k}^{(i)}(\phi(s_i)(k)) -  r_{t, k}^{(i)}(s_i(k))) \\
\leq &~ 3\eps.
\end{align*}
The first step follows from the construction of $\mu = \frac{1}{T}\sum_{t\in [T]}\mu_t = \frac{1}{T}\sum_{t\in [T]}\bigotimes_{i\in [n]}p_t^{(i)}$, the third step follows from the definition of reward $r_{t, k}^{(i)}$ for player $i$ (Line \ref{line:reward} of Algorithm \ref{algo:bayesian}), the last step follows from Lemma \ref{lem:external-swap}.
Hence, we have proved $\mu$ is a (every-type-$3\eps$)-NFCE.

For computation efficiency, a player needs to compute $O(nK)$ reward entries per iteration (Line \ref{line:reward} of Algorithm \ref{algo:bayesian}). Each of those entries can be approximated (with high probability, up to $\pm \eps$ error) by random sampling $O(\log(mTnKL)/\eps^2) = O(\log(mnK/\eps)/\eps^3)$ entries from the joint strategy distribution $\mu_t$ and the prior distribution over $\Theta$ -- we assume one can sample from the conditional distribution $(\rho | \theta_i = k)$ in $O(m)$ time. The total runtime per player equals $O(mnK\log(mnK/\eps)/\eps^3)$.

\end{proof}

\bibliographystyle{alpha}
\bibliography{ref.bib}

\newpage
\appendix
\section{On the rank of CE}
\label{sec:example-app}

We give a simple example of Bayesian game (thus also an example of extensive-form game), where behaviourizing a rank-$2$ CE is no longer a CE. 
There are two players, Alice and Bob, in the game.
Bob has $2$ actions,  one type and his utility is always $0$. 
Alice has $n$ types and $4$ actions. 
For each type, Alice has utility
\begin{align*}
u_A(0, 0) = 0, u_A(1, 0) = 0, u_{A}(2, 0) = 1, u_B(3, 0) =  -2
\end{align*}
and 
\begin{align*}
u_A(0, 1) = 0, u_A(1, 1) = 0, u_{A}(2, 1) = -2, u_A(3, 1) =  1.
\end{align*}

Consider the following rank-$2$ CE $\mu = \frac{1}{2} p_A^{(1)}\otimes  p_B^{(1)} + \frac{1}{2} p_A^{(2)}\otimes p_B^{(2)}$, where $p_A^{(1)} = \frac{1}{3}(0,0,\ldots,0) + \frac{2}{3}(1,1, \ldots, 1)$, $p_A^{(2)} = \frac{2}{3}(0,0,\ldots,0) + \frac{1}{3}(1,1, \ldots, 1)$; $p_B^{(1)} = 0$, $p_B^{(0)} = 1$. 
That is, in the first product distribution, Alice plays all $0$ w.p. $1/3$ and plays all $1$ w.p. $2/3$, Bob always plays $0$; in the second product distribution, Bob plays all $0$ w.p. $2/3$ and plays all $1$ w.p. $2/3$, Bob always plays $1$.
It is not hard to verify that $\mu$ is indeed a CE and Alice receives $0$ utility.

Meanwhile, the only way to behaviourize $\mu$, is to replace $p_A^{(1)}$ with $x_A^{(1)} = (\frac{1}{3}, \frac{2}{3}) \otimes \cdots \otimes (\frac{1}{3}, \frac{2}{3})$ and  $x_B^{(1)} = (\frac{2}{3}, \frac{1}{3}) \otimes \cdots \otimes (\frac{2}{3}, \frac{1}{3})$. Let $\mu'= \frac{1}{2} x_A^{(1)}\otimes  p_B^{(1)} + \frac{1}{2} x_A^{(2)}\otimes p_B^{(2)}$.
Consider the swap function $\phi_A = \{0,1,2,3\}^n \rightarrow \{0,1,2,3\}^n$
\begin{align*}
\phi_A(s) =
\left\{
\begin{matrix}
(2,\ldots, 2) & s\in \{0,1\}^n, |s| \in (n/3 - \sqrt{n}\log(n), n/3 + \sqrt{n}\log(n))\\
(3, \ldots, 3) & s\in \{0,1\}^n, |s| \in (2n/3 - \sqrt{n}\log(n), 2n/3 + \sqrt{n}\log(n))\\
s & \text{otherwise}
\end{matrix}
\right.
\end{align*}
We claim that applying the swap function $\phi_A$, the utility of Alice becomes $1 - o(1)$ (under $\mu'$). 
First, we know that with probability $1-o(1)$, for any $(s_A, s_B) \sim \mu'$, $|s_A| \in (n/3 - \sqrt{n}\log(n), n/3 + \sqrt{n}\log(n))$ or $|s_A| \in (2n/3 - \sqrt{n}\log(n), 2n/3 + \sqrt{n}\log(n))$. In the former case, Bob plays $1$ w.p. $1-o(1)$ so switching to action $(3,3,\ldots, 3)$ has utility $1$; in the latter case, Bob plays $0$ w.p. $1-o(1)$ so switching to action $(2,2,\ldots, 2)$ has utility $1$. This implies $\mu'$ is not a CE.
\section{Missing proof from Section \ref{sec:lower}}
\label{sec:lower-app}
\begin{fact}
\label{fact:product1}
Let $X, Y$ be finite domains. Let $p, q \in \Delta(X\times Y)$ be two probability distributions over $X\times Y$. Define $p(x):= \sum_{y\in Y}p(x, y)$ be the marginal probability of $x \in X$ and $p(y| x) = p(x,y)/p(x)$ be the conditional probability of $y \in Y$ on $x \in Y$.
Suppose $p(x) = q(x)$ holds for any $x \in X$, then 
\[
\TV(p; q) = \sum_{x\in X}p(x) \cdot \TV(p(\cdot | x); q(\cdot| x)).
\]
\end{fact}
\begin{proof}
We have 
\begin{align*}
\TV(p; q) = &~ \frac{1}{2}\sum_{(x, y) \in X\times Y}|p(x, y) - q(x, y)| = \sum_{x\in X}\left(\frac{1}{2}\sum_{y\in Y}|p(x, y) - q(x, y)|\right)\\
= &~ \sum_{x\in X}p(x) \left(\frac{1}{2}\sum_{y\in Y}|p(y|x) - q(y| x)|\right) = \sum_{x\in X}p(x)\cdot \TV(p(\cdot | x); q(\cdot| x)).
\end{align*}
The first step follows from the definition of TV distance, the third step holds since $p(x, y) = p(x)\cdot p(y|x)$, $q(x, y) = q(x)q(y|x) = p(x)q(y|x)$, the last step follows the definition of TV distance.
\end{proof}

\begin{fact}
\label{fact:product2}
Let $X, Y$ be finite domains. Let $w \in \Delta(X \times Y)$ be a joint distribution over $X\times Y$, and $w_{X}, w_{Y}$ be the marginal distribution of $w$ on $X, Y$. 
Let $p = w_{X} \times w_{Y} \in \Delta(X)\times \Delta(Y)$ be a product distribution over $X\times Y$ and it has the same marginal distribution as $w$. 
Then for any product distribution $q \in \Delta(X)\times \Delta(Y)$, we have 
\[
\TV(p; q) \leq 2\TV(w; q)
\]
\end{fact}
\begin{proof}
We have
\begin{align*}
\TV(p;q) \leq  \TV(w_X; q_X) + \TV(w_{Y}; q_{Y}) \leq 2\TV(w; q).
\end{align*}
\end{proof}

\end{document}